\documentclass[final, 10pt, twocolumn, conference]{IEEEtran}
\usepackage{setspace}

\usepackage{amsmath}
\usepackage{amsthm}
\usepackage{amssymb}
\usepackage{cite}
\usepackage{graphicx}
\usepackage{epstopdf}
\usepackage{url}
\usepackage{cite}
\usepackage{texdef2015}
\usepackage{color}
\usepackage{caption}
\usepackage{subcaption}
\usepackage{float}
\allowdisplaybreaks
\usepackage{tikz}
\usetikzlibrary{automata,arrows,positioning,calc,fit,shapes.multipart,chains,shapes}
\usepackage{microtype}

\newtheorem{theorem}{Theorem}
\newtheorem{corollary}{Corollary}

\newtheorem{lemma}{Lemma}

\newcommand{\age}{\Delta}

\newcommand{\negfigspace}{\vspace{-2mm}}

\newcommand{\cf}{c}
\newcommand{\rqr}{r}
\newcommand{\wqr}{w}
\newcommand{\rset}{\mathcal{R}}
\newcommand{\wset}{\mathcal{W}}

\newcommand{\pscale}{\gamma}
\newcommand{\pshape}{\sigma}
\newcommand{\eos}{\mu}

\newcommand{\Esmall}[1]{\textrm{E}[#1]}

\begin{document}
\title{Minimizing Content Staleness in Dynamo-Style Replicated Storage Systems}
 \author{Jing Zhong, Roy D.~Yates and Emina Soljanin \\
 \small 
 Dept. of ECE, Rutgers University, \{jing.zhong, ryates, emina.soljanin\}@rutgers.edu}


\maketitle

\begin{abstract}
Consistency in data storage systems requires any read operation to return the most recent written version of the content.
In replicated storage systems, consistency comes at the price of delay due to large-scale write and read operations.
Many applications with low latency requirements tolerate data staleness in order to provide high availability and low operation latency. 
Using age of information as the staleness metric, we examine a data updating system in which real-time content updates are replicated and stored in a Dynamo-style quorum-based distributed system.
A source sends updates to all the nodes in the system and waits for acknowledgements from the earliest subset of nodes, known as a write quorum.
An interested client fetches the update from another set of nodes, defined as a read quorum.
We analyze the staleness-delay tradeoff in replicated storage by varying the write quorum size.  
With a larger write quorum, an instantaneous read is more likely to get the latest update written by the source. 
However, the age of the content written to the system is more likely to become stale as the write quorum size increases.
For shifted exponential distributed write delay, we derive the age optimized write quorum size that balances the likelihood of reading the latest update and the freshness of the latest update written by the source. 

\end{abstract}

\section{Introduction}
In modern distributed storage systems, data is often replicated across multiple machines or datacenters to support fault tolerance due to server failures, and provide high availability by delivering data through replica servers.
In order to overcome asynchrony in distributed storage systems, {\it quorum}-based algorithms \cite{Attiya1995,Peleg1995,Malkhi1997,Malkhi1998} are well studied and widely used in practice to ensure write and read consistency of replicated data.
In a quorum system, either a write or read of the data goes to a subset of nodes. 
More specifically, the data source or writer sends replicas of the data to all the nodes but only waits for the response from a subset of nodes known as the \emph{write quorum} $\wset$. The client or reader fetches the data from a possibly different subset of nodes, which is called the \emph{read quorum} $\rset$.  
In order to guarantee strict consistency that every read operation returns the most recent written content, a traditional quorum system requires the write quorum $\wset$ and read quorum $\rset$ to overlap in at least one element.
This is known as a \emph{strict quorum}. 
When a quorum is randomly selected by the writer and reader, a strict quorum requires that the write quorum size $\wqr$ and read quorum size $\rqr$ satisfy 
$\wqr + \rqr > n$,
where $n$ is the number of servers/nodes in the system.

However, as the write or read operation to a set of nodes experiences varying random delays, consistency of the storage system comes at the price of delay. 
A wide range of applications have crucial delay requirements, e.g., every 100ms of extra latency cost Amazon $1\%$ in sales, and an extra 0.5s delay in search results cuts Google's traffic by $20\%$ \cite{latencycost}.   
It has been shown that a non-strict or partial quorum of reduced size is widely used in practice because of the latency benefit despite a minor loss in consistency \cite{Bailis2012}.
Amazon's Dynamo database \cite{Dynamo}, and a variety of subsequent database implementations such as Apache Cassandra \cite{Cassandra}, use a non-strict quorum as the data replication mechanism in order to maintain a balance between consistency and latency.

Since strict consistency is not guaranteed in partial quorum systems, the level of consistency is quantified by data \emph{staleness}. 
The definition of data staleness falls into two categories: 1) staleness in time \cite{Golab2011} \cite{Rahman2017} and 2) staleness in data version \cite{Bailis2012}.
In \cite{Golab2011}, a read is considered stale if the value returned was written more than $\delta$ time units before the most recent write, where $\delta$ is a pre-determined threshold. 
In a slightly different time-based staleness definition \cite{Rahman2017}, the data is considered fresh if it was generated no more than $\delta$ time units ago or it's the most recent written data in the system.
On the other hand, \cite{Bailis2012} measures the staleness by how many versions the value returned by a read lags behind the most recent write.

In this work, we characterize data staleness from a strictly time-based perspective.
We examine data monitoring and gathering systems in which real-time content updates generated by a source are stored in a quorum-based distributed database, and a client connects to the database and requests the most recent data update. 
Since the stored data in a node is desired to be as recent as possible, a write operation to a quorum of nodes can be seen as an information update by the source. 
In these applications, the freshness/staleness of the information updates is measured by an ``Age of Information" (AoI)  timeliness metric \cite{Kaul2012infocom,Costa2014,Huang2015,Sun2016,Najm2017}.
If a client reads data by fetching from a set of nodes at some time $t$, and the data has a version time-stamped $u(t)$, then the {\em age} of the data at the client is $t-u(t)$.
We note that the age in time differs from other staleness metrics in distributed storage systems, since we take the age of the most recent written content into account. 
We start by first considering the baseline problem:
in a distributed storage system with $n$ nodes, how does the size of the write quorum $\wqr$ and read quorum $\rqr$ affect the average age of the content returned by a read?

\section{System Model and Metric}\label{sec:system}

\begin{figure}[t]
\centering\small
\begin{tikzpicture}[node distance=1cm]
\node [draw,circle, rounded corners,align=center] (newsource) {Source};
\node[draw,rectangle,rounded corners,align=center,minimum height=4mm,thick] (node_2)[right = 4 of newsource] {Node 2};
\node[draw,rectangle,rounded corners,align=center,minimum height=4mm,thick] (node_1)[above = 0.1 of node_2] {Node 1};
\node[draw,rectangle,rounded corners,align=center,minimum height=4mm,thick] (node_3)[below = 0.1 of node_2] {Node 3};
\node[draw,rectangle,rounded corners,align=center,minimum height=4mm,thick] (node_n)[below = 1.4 of node_2] {Node n};
\draw[->,thick] (newsource.east) -- node[draw,rectangle,minimum height=3mm,thin][above left]{$j$+1} ++(1.5,0) |-  (node_1.west) node[draw,rectangle,minimum height=3mm,thin] [above left = 0 and 0.2]{$j$};
\draw[->,thick] (newsource.east) -- ++(1.5,0) |- (node_3.west)node[draw,rectangle,minimum height=3mm,thin,opacity=.6] [above left = 0 and 0.7]{$j$};
\draw[->,thick] (newsource.east) -- (node_2.west) node[draw,rectangle,minimum height=3mm,thin,opacity=.6] [above left = 0 and 1]{$j$}; 
\draw[->,thick] (newsource.east) -- ++(1.5,0) |- (node_n.west) node[draw,rectangle,minimum height=3mm,thin] [above left = 0 and 0.3]{$j$};
\path (node_3) -- node {$\vdots$}  (node_n);
\node[draw,circle,rounded corners,align=center,] (client)[right = 1 of node_2] {client};
\draw[->,thick] (node_1.east) -- (client);
\draw[->,thick] (node_3.east) -- (client);
\draw[->,thick] (node_n.east) -- (client);
\end{tikzpicture}
\caption{Dynamo-style distributed storage: the source sequentially writes content updates to multiple nodes with random write delays. The next write $j+1$ is initiated right after update $j$ is written to $\wqr$ out of $n$ nodes. The client reads the content through a random set of $\rqr$ nodes and selects the freshest version.}
\label{fig:sysmodel}
\negfigspace
\end{figure}
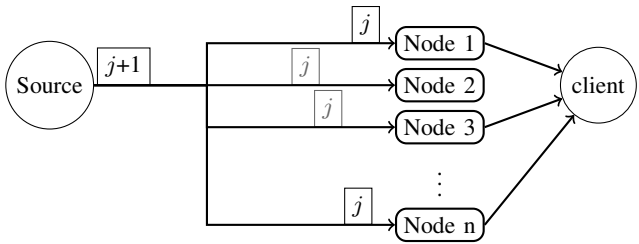

\subsection{Staleness in Dynamo-style Systems}

In a Dynamo-style replicated quorum system, a write/read request by a user will be replicated and sent to \emph{all} the nodes in the system. 
The write operation is considered completed only if the user receives at least $\wqr$ acknowledgements from the $n$ nodes in the system, where $\wqr$ is the pre-determined write quorum size. 
Similarly, the read operation is successful only if the user gets $\rqr$ responses from the system, where $\rqr$ is the read quorum size.
In this paper, we consider a Dynamo-style replicated quorum system with a writing source and a reading client as shown in Fig.~\ref{fig:sysmodel}.
The source writes a time-stamped data item by replicating it to $n$ different nodes in the system with independent random write delays.
An update takes time $X_i$ to be written to node $i$. We refer to $X_i$ as the write delay for node $i$.
We assume that the $X_i$ are i.i.d. shifted exponential $(\lambda,c)$ random variables. Consequently, each $X_i$ has CDF 
\begin{align}
	F_{X}(x)~=~1-e^{-\lambda(x-c)}, \quad x\geq c.
\end{align}
The constant time shift $c>0$ captures the delay produced by the update generation and assembly process.
On the other hand, $\cf$ can also represent a propagation delay on top of an exponential network delay if the source and database are geographically separated.
When freshest received data at time $t$ at node $i$ is time-stamped time $u_i(t)$, the age of information or simply the {\it age} at node $i$, is the random process $\age_i(t)=t-u_i(t)$. 
When a data item reaches node $i$, $u_i(t)$ is advanced to the timestamp of the new content and the node sends an acknowledgement to the source through a feedback channel that is assumed to be instantaneous. 
Assuming that the data item is a time-sensitive content update, the source then obeys a zero-wait policy and initiates a new write request as soon as the current write is completed. 
At the same time, the write operation to the remaining $n-\wqr$ nodes are canceled.

To read data, a client connects to a random set of $\rqr$ nodes in the system, and selects the freshest content among all $\rqr$ nodes.
In \cite{Bailis2012}, the average read latency in Basho Riak \cite{riak}, a commercial distributed database, is shown to be an order of magnitude smaller than the average write latency.
Hence, we assume the read process is instantaneous with zero delay in this work. 
Under this model,  the age at the read client at time $t$ is defined as the minimum age over all the nodes in the read quorum $\rset$, i.e.
\begin{align}
\age(t) = \min_{i\in\rset} \age_i(t). \label{eqn:agedef}
\end{align}
Since the write delays are i.i.d. for each node $i$ and content update $j$, the $\age_i(t)$ processes at each node are statistically identical. The age processes for different read quorum $\rset$ are also statistically identical since $\rset$ is randomly chosen.  
The time average of age process at the client is then given by
\begin{align}
\age = \lim_{\tau\to\infty} \frac{1}{\tau} \int_{t=0}^{\tau} \age(t).
\end{align}

The mathematical model we consider is also relevant to other status updating systems, e.g., multicast with HARQ, where single source transmits coded update packets to multiple clients \cite{Yates2017}.
In related work \cite{Bedewy2017}, update messages are replicated and sent to the receiver through multiple servers; given a general packet arrival process and memoryless packet service times, it was shown that Last-Generated First-Serve scheduling policy is age-optimal. 
In \cite{Sang2017}, a pull-based updating system is considered, in which the arriving source updates are sent to multiple servers and the interested user fetches the update by sending replicated requests to all the servers.
Similar to this work, it was shown there exists an optimal number of responses $k$ from $n$ servers for the user to wait for.
The problem considered in this work differs from \cite{Sang2017} by allowing the source to control when to submit an update based on the delivery feedback.

\subsection{Order Statistics}
We first introduce the notion of order statistics for i.i.d. random variables which plays a key role in our analysis.
We denote the $k$-th order statistic of the random variables $X_1, \ldots, X_n$, i.e., the $k$-th smallest variable, as $X_{k:n}$.

\begin{lemma} {\rm\cite{Arnold2008}}
	For shifted exponential random variable $X$ with CDF $F_{X}(x)~= 1-e^{-\lambda(x-\cf)}, x\geq \cf$., the expectation and variance of the order statistics $X_{k:n}$ are given by 
	\begin{align}
	\E{X_{k:n}} & = \cf+\frac{1}{\lambda}(H_n-H_{n-k}) \label{eqn:os1_exp} \\
	\Var{X_{k:n}} & = \frac{1}{\lambda^2}\left(H_{n^2}-H_{(n-k)^2}\right), \label{osvar_exp}
	\end{align}
	where $H_n$ and $H_{n^2}$ are the generalized harmonic numbers defined as $H_n = \sum_{j=1}^{n}\frac{1}{j}$ and $H_{n^2} = \sum_{j=1}^{n}\frac{1}{j^2}$. 
\end{lemma}

\section{Age Analysis} \label{sec:analysis}

\begin{figure}[t]
\centering
\begin{tikzpicture}[scale=0.22]
\draw [fill=lightgray, ultra thin, dashed] (0,0) to (11,11) to (11,4) to (7,0);
\draw [fill=lightgray, ultra thin, dashed] (17,0) to (25,8) to (25,2) to (23,0);
\draw [<-|] (0,12) node [above] {$\age_{(\wqr,\rqr)}(t)$} -- (0,0) -- (13,0);
\draw [|->] (14,0) -- (28,0) node [right] {$t$};
\draw (0,0) (7,0) -- +(0,-0.5) (17,0) -- +(0,-0.5) (23,0) -- +(0,-0.5);
\draw 
(3,0) node {$\bullet$}  
(7,0) node [below] {$T_1$}
(11,0) node {$\bullet$}
(17,0) node [below] {$T_{j-1}$} 
(20,0) node {$\bullet$} 
(23,0) node [below] {$T_{j}$}
(25,0) node {$\bullet$};
\draw  [<-] (4,2) to [out=60,in=290] (5,7) node [above] {$A_1$};
\draw[<-] (21,2) to [out=60,in=280] (22,7) node [above] {$A_{j}$};
\draw [very thick] (0,4) -- (3,7) -- (3,3)  -- (7,7) 
-- (11,11)-- (11,4) --(13,6);
\draw [very thick] (14,3) -- (20,9) -- (20,3)  -- (25,8) -- (25,2) -- (28,5); 
\draw  [|<->|] (0,10) to node [above] {$Y_{1}$} (7,10);
\draw  [|<->|] (17,10) to node [above] {$Y_{j}$} (23,10);
\end{tikzpicture}
\caption{\small Sample path of the age $\age_{(\wqr,\rqr)}(t)$ with strict quorum $\wqr+\rqr>n$. Update delivery instances are marked by $\bullet$.}
\label{fig:sawtooth_strict}
\end{figure}
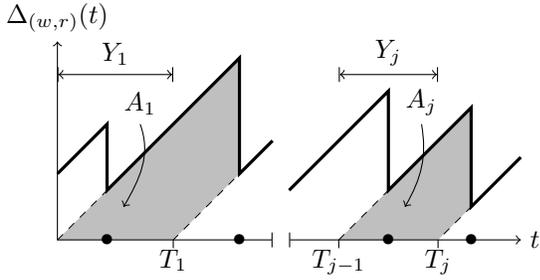  

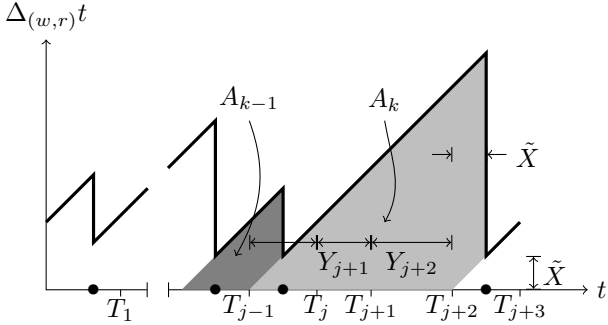
\begin{figure}[t]
\centering
\begin{tikzpicture}[scale=0.09]
\draw [<-|] (0,37) node [above] {$\age_{(\wqr,\rqr)}{t}$} -- (0,0) -- (15,0);
\draw [|->] (18,0) -- (80,0) node [right] {$t$};
\draw (0,0) (11,0) -- +(0,-1) (30,0) -- +(0,-1) (40,0) -- +(0,-1) (48,0) -- +(0,-1) (60,0) -- +(0,-1) (70,0) -- +(0,-1) ;
\fill[lightgray] (30,0) to ++(10,10) to ++(25,25) to ++(0,-30)  to ++(-5,-5);
\fill[gray] (20,0) to (35,15) to (35,5) to (30,0) ;
\draw (7,0) node {$\bullet$}
(25,0) node {$\bullet$}
(35,0) node {$\bullet$}
(65,0) node {$\bullet$};
\draw (11,-6) node [above] {$T_1$} 
(30,-6) node [above] {$T_{j-1}$} (40,-6) node [above] {$T_{j}$} (48,-6) node [above] {$T_{j+1}$} (60,-6) node [above] {$T_{j+2}$} (70,-6) node [above] {$T_{j+3}$};
\draw [very thick] (0,10) -- ++(7,7) -- ++(0,-10) -- ++(8,8);
\draw [very thick]  (18,18) -- ++(7,7)  -- ++(0,-20) -- ++(7,7) -- (35,15) -- ++(0,-10) -- ++(3,3) --  ++(2,2) -- ++(25,25) -- ++(0,-30) -- ++(5,5);
\draw  [|<->|] (30,7) to node [below] {} (40,7);
\draw  [|<->|] (40,7) to node [below] {$Y_{j+1}$} (48,7);
\draw  [|<->|] (48,7) to node [below] {$Y_{j+2}$} (60,7);
\draw [|<-] (60,20) to (57,20);
\draw [|<-] (65,20) to (68,20) node [right] {$\Xtil$};
\draw  [|<->|] (72,0) to node [right] {$\Xtil$} (72,5);
\draw  [<-] (28,5) to [out=60,in=290] (30,25) node [above] {$A_{k-1}$};
\draw[<-] (50,10) to [out=60,in=280] (50,25) node [above] {$A_{k}$};
\end{tikzpicture}
\caption{\small Sample path of the age $\age_{(\wqr,\rqr)}(t)$  with non-strict quorum $\wqr+\rqr\leq n$: successful updates (at times marked by $\bullet$) occur in intervals $1$, $j-1$, $j$, and $j+3$.}
\label{fig:sawtooth_part}
\end{figure}

Our objective is to obtain the average of the age \eqref{eqn:agedef} of the data content given by a read at any time $t$, which is determined by the minimum age over all nodes in the read quorum $\rset$.
Since the read operation is instantaneous by assumption, this is equivalent to evaluating the time-averaged minimum age of any set of $\rqr$ nodes, where $\rqr \in \{1,2,\ldots,n\}$.

With a write quorum of size $\wqr\in\{1,2,\ldots, n\}$, the source considers the current write committed and initiates the next write operation only after the current update is successfully written to $\wqr$ out of $n$ nodes. 
In this case, we denote the write delay for node $i$ and content update $j$ as $X_{ij}$, and the total write delay for update $j$ as $Y_j$.
Note that $Y_j$ is statistically identical to $X_{\wqr:n}$
If one node gets the content earlier than any of the other nodes, it has to wait for an idle period until that content is written to $\wqr-1$ other nodes.
Otherwise, if update $j$ is not written to a node $i$, the node waits for time $Y_j$ until the source starts the next write, and we refer to this random waiting time as a {\it write interval} for $\wqr$ nodes.

For a strict quorum system with $\rqr+\wqr > n$, there is at least one overlapping node between the read quorum $\rset$ and the write quorum $\wset$, thus every update will be eventually written to at least one nodes in the read quorum $\rset$. 
That is, if the client reads the data content right after the most recent committed write, it is guaranteed to receive the most up-to-date update committed by the source.
On the other hand, it is also possible that a read operation returns the most recent write even before the commitment since the client selects the freshest content from all nodes in the read quorum.

Fig. \ref{fig:sawtooth_strict} depicts a sample path of the minimum age of a random set of $\rqr$ nodes which satisfies the strict quorum  requirement. 
The start and end time of the $j$-th write interval are marked as $T_{j-1}$ and $T_{j}=T_{j-1}+Y_{j}$, respectively. 
For update $j$, let's denote the earliest node to complete the write process in the read quorum as node $i$.
The instantaneous age $\age(t)$ drops to exactly the write delay for node $i$, denoted by $X_{ij}$, when the earliest node $i$ receives the update $j$. 
This also implies at time $T_j$ when update $j$ is written to the write quorum $\wset$, the age of the read quorum is $\age(T_i)=Y_j$.

For a non-strict or partial quorum system with $\rqr+\wqr\leq n$, an update by the source may not be successfully read by the client because the random sets $\wset$ and $\rset$ can be disjoint. 
Suppose an update is written to at least one node in the read quorum $\rset$ during write interval $j$ and the next successful write to read quorum $\rset$ is in write interval $j+M$. 
In this case, $M$ is a geometric r.v. with probability mass function (PMF) $P_M(m) = (1-p)^{m-1} p, m\geq 1$. Thus $M$ has first and second moments
\begin{align}
\E{M} & = \frac{1}{p}, \qquad \E{M^2} = \frac{2-p}{p^2}. \label{eqn:EM}
\end{align}
A fresh update fails to be written to the client during a write interval if and only if the write quorum $\wset$ and read quorum $\rset$ do not overlap. 
Let's denote the probability of a write failure as $q=1-p$, then
\begin{align}
q 
&= \Pr\{\rset\cap\wset = \emptyset\} = \begin{cases}
0, & \rqr+\wqr > n \\ 
\binom{n-w}{\rqr}/\binom{n}{\rqr}, & \rqr+\wqr \leq n.
\end{cases}  \label{eqn:q}
\end{align}

A similar example of the age process is shown in Fig. ~\ref{fig:sawtooth_part}. 
We represent the area under the age sawtooth as the concatenation of the polygons $A_1,\ldots,A_k$ as shown in Figs.~\ref{fig:sawtooth_strict} and~\ref{fig:sawtooth_part}. 
The update $j$ is written to the read quorum $\rset$ in the write interval $j$, and the read quorum waits for $M_k=3$ write intervals until the next successful write. 
Note that a strict quorum can be viewed as a special case with deterministic $M_k=1$.
Denote the random variable $\Xtil$ as the write delay of a successful update written to at least one node in the read quorum $\rqr$. Evaluating Fig.~\ref{fig:sawtooth_part} gives the area 
\begin{align}
	A_k = \frac{1}{2} \left(\sum_{l=j}^{j+M_k-1} Y_l + \Xtil_k\right)^2 - \frac{1}{2}\Xtil_k^2. \label{eqn:Ak}
\end{align}
\begin{lemma}
	The average area $A_k$ as shown in Fig. \ref{fig:sawtooth_part} is
	\begin{align*}
		\E{A} 
		&= 	\mathrm{E}\bigl[\Xtil\bigr]\E{M}\E{Y} \nn
		&  \qquad + \frac{1}{2}\E{M^2}(\E{Y})^2 + \frac{1}{2}\E{M}\Var{Y}. 
	\end{align*}
	\label{eqn:EA}
\end{lemma}
\vspace{-5mm}
\begin{proof} Defining  $W=\sum_{l=j}^{j+M_k-1} Y_l$,
	\eqref{eqn:Ak} can be rewritten as 
	\begin{align*}
		A_k = \frac{1}{2} \Bigr[ W^2
        + 2 \Xtil_{k-1}W 
        + \Xtil_{k}^2 \Bigr] - \frac{1}{2}\Xtil_{k}^2.
	\end{align*}
	Since $M_k$ and the $Y_j$  are independent, 
		$\E{W} = \E{M} \E{Y}$.
	It follows that
	\begin{align}
		\E{A_k} = \frac{1}{2} \E{W^2}
        + \Esmall{\Xtil} \E{M} \E{Y}. \label{eqn:EAproof}
	\end{align}
	 The random sum of random variables $W$ has second moment 
	\begin{align*}
		\E{W^2} & = (\E{W})^2 + \Var{W} \nn
		& = (\E{M})^2 (\E{Y})^2 + \E{M}\Var{Y} + \Var{M}(\E{Y})^2 \nn
		& = \E{M^2} (\E{Y})^2 + \E{M}\Var{Y}.
	\end{align*}
	Substituting $\E{W^2}$
    back into \eqref{eqn:EAproof} completes the proof.
\end{proof}

It follows from Fig. \ref{fig:sawtooth_part} that the average age is given by
\begin{align}
\age_{(\wqr,\rqr)} = \frac{\E{A}}{\E{M}\E{Y}}. \label{eqn:age_raw}
\end{align}

\begin{theorem}
Consider a Dynamo-style $n$-node quorum system with write quorum size $\wqr$ and read quorum size $\rqr$. The source sequentially writes content updates to the system. 
Assuming the freshest content is selected from the read quorum and the read operation is instantaneous, the average age of the content observed by the client is 
\begin{enumerate}
\item for $\wqr+\rqr >n$,
\begin{align*}
\age_{(\wqr,\rqr)} & = \sum_{i=1}^{\wqr} \E{X_{i:n}} \frac{\binom{n-i}{\rqr-1}}{\binom{n}{\rqr}}
+ \frac{1}{2} \frac{ \E{X^2_{\wqr:n}}}{\E{X_{\wqr:n}}};
\end{align*}
\item for $\wqr+\rqr \leq n$,
\begin{align*}
\age_{(\wqr,\rqr)} & = \sum_{i=1}^{\wqr} \E{X_{i:n}} \frac{\binom{n-i}{\rqr-1}}{\binom{n}{\rqr}-\binom{n-\wqr}{\rqr}} \nn
& \quad + \frac{1}{2} \frac{\binom{n}{\rqr}+\binom{n-\wqr}{\rqr}}{\binom{n}{\rqr}-\binom{n-\wqr}{\rqr}} \E{X_{\wqr:n}} + \frac{1}{2} \frac{ \Var{X_{\wqr:n}}}{\E{X_{\wqr:n}}}. 
\end{align*}
\end{enumerate}
\label{thm:age_w_r}
\end{theorem}

\begin{proof}
Substituting Lemma \ref{eqn:EA} into \eqref{eqn:age_raw} yields 
\begin{align}
\age_{(\wqr,\rqr)} & = \Esmall{\Xtil} +   \frac{\E{M^2}}{2\E{M}} \E{Y} + \frac{1}{2} \frac{\Var{Y}}{\E{Y}} , \label{eqn:age_w_r_raw}
\end{align}
where $M$ is the number of write intervals between successful writes to the read quorum. 
With $p=1-q$, it follows from (\ref{eqn:EM}) and (\ref{eqn:q}) that 
\begin{align}
\frac{\E{M^2}}{2\E{M}} 
&= \frac{1+q}{2(1-q)} = \begin{cases}
\frac{1}{2}, & \rqr+\wqr > n \\ 
\frac{1}{2} \frac{\binom{n}{\rqr}+\binom{n-\wqr}{\rqr}}{\binom{n}{\rqr}-\binom{n-\wqr}{\rqr}}, & \rqr+\wqr \leq n.
\end{cases} \label{eqn:EM_r}
\end{align}
Thus, for strict quorum $\wqr+\rqr>n$, \eqref{eqn:age_w_r_raw} can be written as
\begin{align}
	\age_{(\wqr,\rqr)} & = \Esmall{\Xtil} +   \frac{1}{2} \E{Y} + \frac{1}{2} \frac{\Var{Y}}{\E{Y}} \nn
	& =  \Esmall{\Xtil} +  \frac{1}{2} \frac{\E{Y^2}}{\E{Y}}. \label{eqn:age_w_r_rawstrict}
\end{align}
Denote the node $\underline{i}_\rset$ as the node with least write delay in the read quorum $\rset$, i.e.,
\begin{align*}
\underline{i}_\rset = \mathrm{arg} \min_{i\in\rset} X_i.
\end{align*}
In addition, we rewrite the write quorum set as
\begin{align}
	 \wset = \{i_1,i_2,\ldots,i_\wqr\},
\end{align}
where $X_{i_k} = X_{k:n}$ is the $k$-th smallest write delay in the write quorum.

For strict quorum $\wqr+\rqr>n$, the average write delay for a successful update read by the client is given by
\begin{subequations}
\begin{align}
\E{\Xtil} 
&= \E{X_{\underline{i}_\rset}\,|\,\underline{i}_\rset\in\wset} \label{eqn:EXtil0} \\
&= \sum_{k=1}^{\wqr} \E{X_{k:n}} \Pr\{\underline{i}_\rset=i_k \; | \; \underline{i}_\rset \in \wset\}, \label{eqn:EXtil1} \\
&= \sum_{k=1}^{\wqr} \E{X_{k:n}}\frac{\Pr[\underline{i}_\rset=i_k]}{1-q}, \label{eqn:EXtil2} \\
&= \sum_{k=1}^{\wqr} \E{X_{k:n}} \frac{\binom{n-k}{\rqr-1}}{\binom{n}{\rqr}}. \label{eqn:EXtil3}
\end{align}
\end{subequations}
In (\ref{eqn:EXtil0}), the expectation of $\Xtil$ is defined as the expectation of the minimum of all the write delays $X_i$ in the read quorum $\rset$, given that this minimum is also in the write quorum $\wset$. 
(\ref{eqn:EXtil1}) is obtained by averaging over the conditional expectation of all possible order statistics $X_{k:n}$.
And the complementary event of the condition $\underline{i}_\rset\in\wset$ is that both subsets do not overlap, $\rset\cap\wset = \emptyset$, which yields (\ref{eqn:EXtil2}).
From (\ref{eqn:EXtil2}) to (\ref{eqn:EXtil3}), it follows from \eqref{eqn:q} that $q=0$. 
And we have $\Pr[\underline{i}_\rset =i_k] = \binom{n-k}{r-1}/\binom{n}{r}$, since $X_{k:n}$ is smallest in the read quorum $\rset$ with size $\rqr$, and the remaining $r-1$ values are randomly chosen from the subset $\{ X_{k+1:n}, \ldots, X_{n:n}\}$.

Similarly, for the non-strict quorum with $\wqr+\rqr\leq n$, 
\begin{subequations}
\begin{align}
\Esmall{\Xtil} 
&= \E{X_{\underline{i}_\rset}\,|\,\underline{i}_\rset\in\wset} \\
&= \sum_{k=1}^{n-\rqr+1} \E{X_{k:n}} \Pr\{\underline{i}_\rset=i_k \; | \; \underline{i}_\rset \in \wset\},\\
&= \sum_{k=1}^{n-\rqr+1} \E{X_{k:n}} \frac{\Pr[\underline{i}_\rset=i_k]}{1-q}, \label{eqn:EXtilnon2} \\
&= \sum_{k=1}^{n-\rqr+1} \E{X_{k:n}} \frac{\binom{n-k}{\rqr-1}}{\binom{n}{\rqr}-\binom{n-\wqr}{\rqr}}. \label{eqn:EXtilnon3}
\end{align}	
\end{subequations}
From \eqref{eqn:EXtilnon2} to \eqref{eqn:EXtilnon3}, we apply $q=\binom{n-w}{r}/\binom{n}{r} $ in \eqref{eqn:q} for the case $\rqr+\wqr\leq n$. 
Note that the length of a write interval is $Y=X_{\wqr:n}$. 
To get Theorem \ref{thm:age_w_r}, we substitute (\ref{eqn:EXtil3}) back to  \eqref{eqn:age_w_r_rawstrict} for $\wqr+\rqr>n$, and substitute (\ref{eqn:EM_r}) and \eqref{eqn:EXtilnon3} back to (\ref{eqn:age_w_r_raw}) for $\wqr+\rqr\leq n$.

\end{proof}

\begin{corollary}
Let $\beta=1-\alpha=1-\wqr/n$.
For shifted exponential $(\lambda,c)$ write delay $X$ and a given read quorum size $\rqr$, the average age at the client can be approximated for large $n$ as:
\begin{enumerate}
\item for $\wqr+\rqr >n$,
\begin{align*}
\age_{(\wqr,\rqr)} & \approx 
\frac{1-2\beta^\rqr}{2\lambda} \log\frac{1}{\beta} + (1-\beta^\rqr) (\cf+\frac{1}{\lambda\rqr}) + \frac{\cf}{2}. 
\end{align*}
\item for $\wqr+\rqr \leq n$,
\begin{align*}
\age_{(\wqr,\rqr)} & \approx
\frac{1}{\lambda r} + \frac{1}{2\lambda}\log\frac{1}{\beta} + c + \frac{c(1+\beta^r)}{2(1-\beta^r)}.
\end{align*}
\end{enumerate} \label{thm:wr_approx}
\end{corollary}

\begin{corollary} \label{thm:wr_optimal}
Denote $\omega=\beta^r$, the optimal $\omega^*$ that minimizes Corollary \ref{thm:wr_approx} for positive $\lambda$ and $\cf$ is
\begin{align}
\omega^* = (\lambda\cf\rqr+1)-\sqrt{(\lambda\cf\rqr+1)^2-1}.
\end{align} 
\end{corollary}
Proofs for both corollaries are provided in the appendix.

\section{Evaluation} \label{sec:evaluation}

\begin{figure}[pt]
\centering
\begin{subfigure}[b]{0.5\textwidth}
\centering
\includegraphics[width=\textwidth]{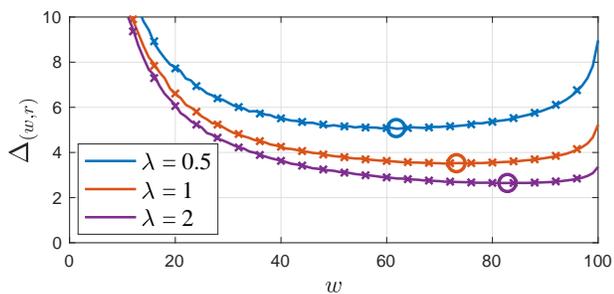}
\caption{read quorum $\rqr=1$, $n=100$.}
\label{fig:r1}
\end{subfigure}
\begin{subfigure}[b]{0.5\textwidth}
\centering
\includegraphics[width=\textwidth]{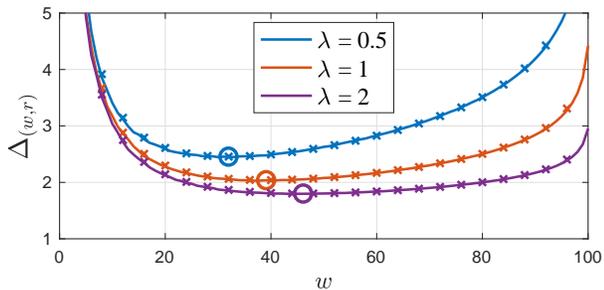}
\caption{read quorum $\rqr=5$, $n=100$.}
\label{fig:r5}
\end{subfigure}
\begin{subfigure}[b]{0.5\textwidth}
\centering
\includegraphics[width=\textwidth]{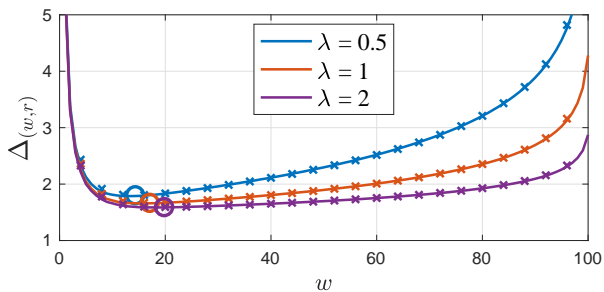}
\caption{read quorum $\rqr=20$, $n=100$.}
\label{fig:r20}
\end{subfigure}
\caption{Average age as a function of write quorum $\wqr$ for shifted exponential write delay. $\times$ marks the approximate average age, and $\circ$ marks the minimized approximate age $\hat{\age}(\wqr^*)$.} 
\label{fig:SExp}
\negfigspace
\end{figure}

Figure \ref{fig:SExp} compares the simulation results of the average age $\age_{(\wqr,\rqr)}$ as a function of the write quorum size $\wqr$ given different pre-determined read quorum $\rqr$. 
In this experiment, the total number of nodes $n=100$, and every write to a node has shifted exponential delay with $\cf=1$ and different $\lambda$. 
The approximation in Corollary~\ref{thm:wr_approx} is marked with $\times$, and the near-optimal write quorum size $\wqr^*$ in Corollary~\ref{thm:wr_optimal} is marked with $\circ$. 
By looking at the three curves in a single figure, We observe that for a given $c$ and $n$, the optimal write quorum size $\wqr^*$ increases as the exponential rate $\lambda$ increases.  
Comparing across figures with different read quorum $\rqr$, we observe that the average age decreases as $\rqr$ increases, and the optimal write quorum size $\wqr$ also decreases significantly.
For example, the client reads the data from only a single node in Fig. \ref{fig:r1}. In this case, in order to minimize the staleness of the content update, the source should consider a write operation to be complete once the update is written to 60 out of 100 nodes when $\lambda=0.5$.
However, if the client is connected to 5 nodes as shown in Fig. \ref{fig:r5}, it is best for the source to only write to around 30 nodes.
This observation also implies that choosing a partial/non-strict quorum is usually a winning strategy to minimize the content staleness in replicated storage system. 
Although strict quorum guarantees that any read after the write commit can return the most recent write, the content written to the system is more likely to become stale due to larger write delay.

\section{Conclusion} \label{sec:conclusion}
We have examined a real-time data replicated storage system in which content updates are replicated and stored in a dynamo-style quorum systems.
Either a write or read request goes to all the nodes in the system, and it is considered complete once there are at least $\wqr$ or $\rqr$ responses.
Assuming that the write delay dominates the latency, the freshness of the replicated storage system is measured by the average age of the content returned by a read at any time $t$.
As the write quorum size $\wqr$ increases, an instantaneous read from a read quorum with predetermined size $\rqr$ is more likely to get the latest version generated by the source. 
However, the age of the content also increases as the write quorum size $\wqr$ increases.
We have derived the optimal $\wqr$ given a read quorum size $\rqr$ such that the average age of the content returned by a read is minimized, and showed by experiment that the optimal $\wqr$ satisfies the non-strict quorum $\wqr+\rqr\leq n$.

The analysis presented in this work is based on the assumption that the read delay is negligible compared to the write delay, and the feedback channels from the nodes are instantaneous. 
We are also aware of more general cases where the read delay is also significant and random such that the content may become stale during the read process.
Under this scenario, a different write quorum $\wqr$ should be chosen to deal with the possibility of stale data due to the read delay.

\section*{Acknowledgment}
\addcontentsline{toc}{section}{Acknowledgment}
Part of this research is based upon work supported by the National Science Foundation under grant CNS-1422988.

\bibliographystyle{IEEEtran}
\bibliography{ref}

\begin{thebibliography}{10}
\providecommand{\url}[1]{#1}
\csname url@samestyle\endcsname
\providecommand{\newblock}{\relax}
\providecommand{\bibinfo}[2]{#2}
\providecommand{\BIBentrySTDinterwordspacing}{\spaceskip=0pt\relax}
\providecommand{\BIBentryALTinterwordstretchfactor}{4}
\providecommand{\BIBentryALTinterwordspacing}{\spaceskip=\fontdimen2\font plus
\BIBentryALTinterwordstretchfactor\fontdimen3\font minus
  \fontdimen4\font\relax}
\providecommand{\BIBforeignlanguage}[2]{{%
\expandafter\ifx\csname l@#1\endcsname\relax
\typeout{** WARNING: IEEEtran.bst: No hyphenation pattern has been}%
\typeout{** loaded for the language `#1'. Using the pattern for}%
\typeout{** the default language instead.}%
\else
\language=\csname l@#1\endcsname
\fi
#2}}
\providecommand{\BIBdecl}{\relax}
\BIBdecl

\bibitem{Attiya1995}
H.~Attiya, A.~Bar-Noy, and D.~Dolev, ``{Sharing Memory Robustly in
  Message-Passing Systems.}'' \emph{J.\ ACM}, pp. 1--12, 1995.

\bibitem{Peleg1995}
D.~Peleg and A.~Wool, ``The availability of quorum systems,'' \emph{Information
  and Computation}, vol. 123, no.~2, pp. 210--223, 1995.

\bibitem{Malkhi1997}
D.~Malkhi, M.~Reiter, and R.~Wright, ``Probabilistic quorum systems,'' in
  \emph{Proceedings of ACM Symposium on Principles of Distributed Computing},
  1997, pp. 267--273.

\bibitem{Malkhi1998}
D.~Malkhi and M.~Reiter, ``Byzantine quorum systems,'' \emph{Distributed
  Computing}, vol.~11, no.~4, pp. 203--213, 1998.

\bibitem{latencycost}
H.~Scalability, ``{Latency is everywhere and it casts you sales - How to crush
  it},'' 2009,
  {http://highscalability.com/latency-everywhere-and-it-costs-you-sales-how-crush-it}.

\bibitem{Bailis2012}
P.~Bailis, S.~Venkataraman, M.~J. Franklin, J.~M. Hellerstein, and I.~Stoica,
  ``{Probabilistically Bounded Staleness for Practical Partial Quorums.}''
  \emph{Proceeding of VLDB}, 2012.

\bibitem{Dynamo}
G.~DeCandia, D.~Hastorun, M.~Jampani, G.~Kakulapati, A.~Lakshman, A.~Pilchin,
  S.~Sivasubramanian, P.~Vosshall, and W.~Vogels, ``Dynamo: amazon's highly
  available key-value store,'' \emph{ACM SIGOPS operating systems review},
  vol.~41, no.~6, pp. 205--220, 2007.

\bibitem{Cassandra}
A.~Lakshman and P.~Malik, ``Cassandra: A decentralized structured storage
  system,'' \emph{SIGOPS Oper. Syst. Rev.}, vol.~44, no.~2, Apr. 2010.

\bibitem{Golab2011}
W.~Golab, X.~Li, and M.~A. Shah, ``{Analyzing Consistency Properties for Fun
  and Profit},'' in \emph{ACM SIGACT-SIGOPS Symposium on Principles of
  Distributed Computing(PODC)}, New York, NY, USA, 2011, pp. 197--206.

\bibitem{Rahman2017}
M.~R. Rahman, L.~Tseng, S.~Nguyen, I.~Gupta, and N.~Vaidya, ``{Characterizing
  and Adapting the Consistency-Latency Tradeoff in Distributed Key-Value
  Stores},'' \emph{ACM Transactions on Autonomous and Adaptive Systems},
  vol.~11, no.~4, pp. 1--36, 2017.

\bibitem{Kaul2012infocom}
S.~Kaul, R.~D. Yates, and M.~Gruteser, ``Real-time status: How often should one
  update?'' in \emph{Proc.\ INFOCOM}, Apr. 2012, pp. 2731--2735.

\bibitem{Costa2014}
M.~Costa, M.~Codreanu, and A.~Ephremides, ``{Age of information with packet
  management},'' in \emph{Proc. IEEE Int.\ Symp.\ Inform.\ Theory}, 2014, pp.
  1583--1587.

\bibitem{Huang2015}
L.~Huang and E.~Modiano, ``{Optimizing age-of-information in a multi-class
  queueing system},'' in \emph{Proc. IEEE Int.\ Symp.\ Inform.\ Theory}, Jun.
  2015, pp. 1681--1685.

\bibitem{Sun2016}
Y.~Sun, E.~Uysal-Biyikoglu, R.~Yates, C.~E. Koksal, and N.~B. Shroff, ``Update
  or wait: How to keep your data fresh,'' in \emph{Proc.\ INFOCOM}, 2016.

\bibitem{Najm2017}
E.~Najm, R.~D. Yates, and E.~Soljanin, ``Status updates through m/g/1/1 queues
  with harq,'' in \emph{Proc. IEEE Int.\ Symp.\ Inform.\ Theory}, 2017.

\bibitem{riak}
{Basho Riak}, {http://www.basho.com/products/{\#}riak}.

\bibitem{Yates2017}
R.~D. Yates, E.~Najm, E.~Soljanin, and J.~Zhong, ``Timely updates over an
  erasure channel,'' in \emph{Proc. IEEE Int.\ Symp.\ Inform.\ Theory}, 2017.

\bibitem{Bedewy2017}
A.~M. Bedewy, Y.~Sun, and N.~B. Shroff, ``Minimizing the age of the information
  through queues,'' \emph{arXiv preprint arXiv:1709.04956}, 2017.

\bibitem{Sang2017}
Y.~Sang, B.~Li, and B.~Ji, ``The power of waiting for more than one response in
  minimizing the age-of-information,'' \emph{arXiv preprint arXiv:1704.04848},
  2017.

\bibitem{Arnold2008}
B.~C. Arnold, N.~Balakrishnan, and H.~N. Nagaraja, \emph{A first course in
  order statistics}.\hskip 1em plus 0.5em minus 0.4em\relax SIAM, 2008.

\end{thebibliography}

\appendix
\subsubsection*{\bf Proof of Corollary \ref{thm:wr_approx}}

For shifted exponential r.v. $X$, 
\begin{align}
\frac{ \E{X^2_{\wqr:n}}}{\E{X_{\wqr:n}}} & = \E{X_{\wqr:n}} + \frac{ \Var{X_{\wqr:n}}}{\E{X_{\wqr:n}}} \nn
&= \E{X_{\wqr:n}}+\frac{H_{n^2}-H_{(n-\wqr)^2}}{2\lambda^2\cf + 2\lambda (H_n-H_{n-\wqr})}.
\end{align}
Note that the sequence $H_{n^2}$ is monotonically increasing and $\lim_{n\to\infty} H_{n^2} = \pi^2/6$, thus $H_{n^2}-H_{(n-\wqr)^2}$ is negligible and 
\begin{align}
\lim_{n\to\infty} \frac{\Var{X_{\wqr:n}}}{\E{X_{\wqr:n}}} = 0, \label{eqn:var_over_E}
\end{align}
It follows from \eqref{eqn:var_over_E} that 
\begin{align}
	\lim_{n\to\infty}  \frac{ \E{X^2_{\wqr:n}}}{\E{X_{\wqr:n}}} =  \E{X_{\wqr:n}}. \label{eqn:E2-E}
\end{align}
With large $n$, we also approximate the harmonic number by $H_i \approx \log i + \gamma $, thus
\begin{align}
	\E{X_{\wqr:n}} & \approx \cf + \frac{1}{\lambda} \left(\log n - \log(n-\wqr)\right) \nn
	& = \cf + \frac{1}{\lambda} \left(\log \frac{n}{n-\wqr} \right). \label{eqn:os1appr_exp}
\end{align} 
Note that this approximation only holds when $\wqr<n$.

Let's substitute \eqref{eqn:var_over_E} and \eqref{eqn:os1appr_exp} into Theorem \ref{thm:age_w_r}, and approximate the binomial coefficient by 
$\binom{n}{k}\approx \frac{n^k}{k!}$. 	
For $\wqr+\rqr>n$, Theorem \ref{thm:age_w_r} is then rewritten as
	\begin{subequations}
	\begin{align}
		\age & \approx \sum_{i=1}^{\wqr} \E{X_{i:n}} \frac{(n-i)^{\rqr-1}\rqr}{n^\rqr} + \frac{1}{2} \E{X_{\wqr:n}} \\
		& \approx \sum_{i=1}^{\wqr} \left( \frac{1}{\lambda} \log\left(\frac{n}{n-i}\right) + \cf \right) \frac{(n-i)^{\rqr-1}\rqr}{n^\rqr} \label{eqn:throwE1} \nn
		& \qquad \qquad + \frac{1}{2\lambda}\log\left(\frac{n}{n-\wqr}\right) + \frac{\cf}{2} \\
		& \approx r \int_{x=0}^{\alpha=\frac{\wqr}{n}} \left( \frac{1}{\lambda} \log\left(\frac{1}{1-x}\right) + \cf \right) (1-x)^{\rqr-1}  \mathrm{d}x \nn
		& \qquad \qquad + \frac{1}{2\lambda}\log\left(\frac{1}{1-\alpha}\right) + \frac{\cf}{2} \label{eqn:sum_to_int} \\
		& = \frac{\Big(1-(1-\alpha)^\rqr(1-\rqr\log(1-\alpha))\Big)}{\lambda\rqr} + \cf\left(1-(1-\alpha)^\rqr\right) \nn
		& \qquad \qquad + \frac{1}{2\lambda}\log\left(\frac{1}{1-\alpha}\right) + \frac{\cf}{2} \\
		& = \frac{1-2(1-\alpha)^\rqr}{2\lambda} \log\frac{1}{1-\alpha} \nn
		& \qquad \qquad + (1-(1-\alpha)^\rqr) (\cf+\frac{1}{\lambda\rqr}) + \frac{\cf}{2}. 
	\end{align}			
	\end{subequations}
	In \eqref{eqn:throwE1}, we use the limit in \eqref{eqn:E2-E} as an approximate.
	In \eqref{eqn:sum_to_int}, we denote $\alpha=\wqr/n$ and approximate the sum $\sum_{i=1}^{\wqr} f(i)$ by the integral $\int_{i=0}^\wqr f(i) \mathrm{d}i $. \vspace{1mm} \\ 
For $\wqr+\rqr\leq n$, 
	\begin{subequations}
	\begin{align}
		\age & \approx \sum_{i=1}^{\wqr} \E{X_{i:n}} \frac{(n-i)^{\rqr-1}\rqr}{n^\rqr-(n-\wqr)^\rqr} \nn
		& \quad + \frac{n^\rqr+(n-\wqr)^\rqr}{2(n^\rqr-(n-\wqr)^\rqr)} \E{X_{\wqr:n}} \label{eqn:throwE2} \\
		& \approx \sum_{i=1}^{\wqr} \left( \frac{1}{\lambda} \log\left(\frac{n}{n-i}\right) + \cf \right) \frac{(n-i)^{\rqr-1}\rqr}{n^\rqr-(n-\wqr)^\rqr} \nn
		& \quad + \frac{n^\rqr+(n-\wqr)^\rqr}{2(n^\rqr-(n-\wqr)^\rqr)} \left( \frac{1}{\lambda}\log\left(\frac{n}{n-\wqr}\right) + \cf \right) \\
		& \approx \int_{x=0}^{\alpha} \left( \frac{1}{\lambda} \log\left(\frac{1}{1-x}\right) + \cf \right) \frac{(1-x)^{\rqr-1}\rqr}{1-(1-\alpha)^\rqr}  \mathrm{d}x \nn
		& \quad + \frac{1+(1-\alpha)^\rqr}{2(1-(1-\alpha)^\rqr)} \left( \frac{1}{\lambda}\log\left(\frac{1}{1-\alpha}\right) + \cf \right) \\	
		& = \frac{1}{\lambda\cf} - \frac{(1-\alpha)^\rqr}{\lambda(1-(1-\alpha)^\rqr)}  \log \left( \frac{1}{1-\alpha} \right) + c \nn
		& \quad +  \frac{(1+(1-\alpha)^\rqr)}{2\lambda(1-(1-\alpha)^\rqr)}  \log \left( \frac{1}{1-\alpha} \right) + \frac{\cf(1+(1-\alpha)^\rqr)}{2(1-(1-\alpha)^\rqr)} \\
		& = \frac{1}{\lambda \rqr} + \frac{1}{2\lambda}\log \left( \frac{1}{1-\alpha} \right) + c + \frac{c(1+(1-\alpha)^\rqr)}{2(1-(1-\alpha)^\rqr)}.
	\end{align}	
	\end{subequations}	
To obtain \eqref{eqn:throwE2}, we use the limit in \eqref{eqn:var_over_E} as an approximate and substitute it back to Theorem \ref{thm:age_w_r}. To simplify the expression we further denote $\beta=1-\alpha$ to complete the proof.
\vspace{1mm}
\subsubsection*{\bf Proof of Corollary \ref{thm:wr_optimal}}
We first prove by contradiction that the optimal $\beta$ doesn't fall into strict quorum region. For $\wqr+\rqr>n$, taking the derivative of the approximation 1) in corollary \ref{thm:wr_approx} gives
\begin{align}
	\frac{\mathrm{d} \age}{\mathrm{d} \beta} & =
	\frac{\beta^{\rqr-1}(\lambda\cf\rqr-\rqr\log\beta)}{\lambda}. \label{eqn:deriv1}
\end{align}
Thus we have the optimal $\beta^* = e^{\lambda\cf}$ by setting \eqref{eqn:deriv1} to zero. Since $\lambda$ and $\cf$ are positive, $e^{\lambda\cf}>1$ contradicts $\beta\in(0,1)$.
For non-strict quorum $\wqr+\rqr\leq n$, we let the derivative of the approximation 2) to be zero, i.e.,
\begin{align}
	\frac{\mathrm{d} \age}{\mathrm{d} \beta} & = \frac{\beta^{2\rqr} - 2(\lambda\cf\rqr+1)\beta^\rqr +1}{\lambda\beta(\beta^\rqr-1)} = 0
\end{align}
Since $\beta\in(0,1)$, it is equivalent that
\begin{align}
	\beta^{2\rqr} - 2(\lambda\cf\rqr+1)\beta^\rqr +1 & = 0. \label{eqn:deriv}
\end{align}
We define $\omega=\beta^\rqr$, and the solution to \eqref{eqn:deriv} is given by
\begin{align}
	\omega^* = (\lambda\cf\rqr+1)-\sqrt{(\lambda\cf\rqr+1)^2-1}.
\end{align}

\end{document}